\documentclass[10pt,english,conference]{IEEEtran}
\IEEEoverridecommandlockouts 
\IEEEpubid{\hspace{-\columnwidth}978-1-4799-6773-5/14/\$31.00 ©2014 IEEE}

\usepackage{amsthm,amsmath,amssymb}
\usepackage{graphicx}
\usepackage{cite}
\usepackage{framed}

\usepackage{tikz}
\usetikzlibrary{arrows,matrix}

\newcommand{\ie}{i.e.}

\newcommand{\eg}{e.g.}

\renewcommand{\S}{Section~}

\newcommand{\mc}[1]{\mathcal{#1}}

\newcommand{\mbb}[1]{\mathbb{#1}}

\newcommand{\vect}[1]{\boldsymbol{#1}}
\newcommand{\vectm}[1]{\mathbf{#1}}

\newcommand{\set}[1]{\mathcal{#1}}
\newcommand{\compl}[1]{\overline{#1}}

\newcommand{\transpose}{\top}
\newcommand{\normop}[1]{ \Upsilon(#1) }

\newcommand{\inner}[2]{\left\langle #1,#2 \right\rangle}

\DeclareMathOperator{\ind}{ind}
\DeclareMathOperator{\lind}{lind}
\DeclareMathOperator{\vlind}{\overrightarrow{\lind}}
\DeclareMathOperator{\aut}{Aut}
\DeclareMathOperator{\Span}{Span}
\DeclareMathOperator{\Hom}{Hom}
\DeclareMathOperator{\core}{Core}
\DeclareMathOperator{\minrank}{minrank}

\newtheorem{lemma}{Lemma}
\newtheorem{theorem}{Theorem}
\newtheorem{corollary}{Corollary}

\newtheorem{definition}{Definition}

\newtheorem{example}{Example}
\newtheorem{construction}{Construction}

\author{\IEEEauthorblockN{Javad B. Ebrahimi\IEEEauthorrefmark{1}, Mahdi~Jafari~Siavoshani\IEEEauthorrefmark{2}}
\IEEEauthorblockA{\IEEEauthorrefmark{1}Institute of Network Coding,
Chinese University of Hong Kong, Hong Kong\\
\IEEEauthorrefmark{2}Computer Engineering Department, Sharif University of Technology, Tehran, Iran\\
Email: \text{javad@inc.cuhk.edu.hk}, \text{mjafari@sharif.edu}}
\thanks{The work described in this paper was partially supported by a grant from University Grants Committee of the Hong Kong Special Administrative Region, China (Project No. AoE/E-02/08).}
}

\title{Linear Index Coding via Graph Homomorphism}

\begin{document}
\maketitle

\begin{abstract}
In \cite{YosBirkJayKol-IT11,ChlamHaviv-CoRR11} it is shown that the minimum broadcast rate of a linear index code over a finite field $\mathbb{F}_q$ is equal to an algebraic invariant of the underlying digraph, called $\minrank_q$. In \cite{PeetersRene-Combinatorica96}, it is proved that for $\mathbb{F}_2$ and any positive integer $k$, $\minrank_q(G) \leq k$ if and only if there exists a homomorphism from the complement of the graph $G$ to the complement of a particular undirected graph family called ``graph family $\{G_k\}$''. As observed in \cite{ChlamHaviv-CoRR11}, by combining these two results one can relate the linear index coding problem of undirected graphs to the graph homomorphism problem. In \cite{EbrahimiJafari-ITW14}, a direct connection between linear index coding problem and graph homomorphism problem is introduced. In contrast to the former approach, the direct connection holds for digraphs as well and applies to any field size. More precisely, in \cite{EbrahimiJafari-ITW14}, a graph family $\{H_k^q\}$ has been introduced and shown that whether or not the scalar linear index of a digraph $G$ is less than or equal to $k$ is equivalent to the existence of a graph homomorphism from the complement of $G$ to the complement of $H_k^q$.

In this paper, we first study the structure of the digraphs $H_k^q$ defined in \cite{EbrahimiJafari-ITW14}. Analogous to the result of \cite{ChlamHaviv-CoRR11} about undirected graphs, we prove that $H_k^q$'s are vertex transitive digraphs. Using this, and by applying a lemma of Hell and Nesetril \cite{HellNese-Book04-GraphHomomorphism}, we derive a class of necessary conditions for digraphs $G$ to satisfy $\lind_q(G)\leq k$.  Particularly, we obtain new lower bounds on $\lind_q(G)$. 


%
%

Our next result is about the computational complexity of scalar linear index of a digraph. It is known that deciding whether the scalar linear index of an undirected graph is equal to $k$ or not is NP-complete for $k\ge 3$ and is polynomially decidable for $k=1,2$ \cite{PeetersRene-Combinatorica96}. 
For digraphs, it is shown in \cite{DauSkaChe-IT14} that for the binary alphabet, the decision problem for $k=2$ is NP-complete. We use graph homomorphism framework to extend this result to arbitrary alphabet.
\end{abstract}


\section{Introduction}\label{sec:Introduction}
The \emph{index coding} problem, first introduced by Birk and Kol in the context of satellite communication \cite{BirkKol-IndexCoding-INFOCOM98}, has received significant attention during past years (see for example \cite{AlonLubetStavWeinHassid-Focs08, LubetStav-IT09, RouaySprintGeorgh-IT10, BlaKleLub-CoRR10, YosBirkJayKol-IT11, BerLangberg-ISIT11, HavivLangberg-ISIT12, TehraniDimakisNeely-ISIT12, MalCadJafar-CoRR12, ShanmugamDimakisLangberg-CoRR13, ArbabBanderKimSasogluWang-ISIT13}). This problem has many applications such as satellite communication, multimedia distribution over wireless networks, and distributed caching. Despite its simple description, the index coding problem has a rich structure and it has intriguing connections to some information theory problems. It has been recently shown that the feasibility of any network coding problem can be reduce to an equivalent feasibility problem in the index coding problem (and vice versa) \cite{EffrosSalimLangberg-CoRR12}. Also an interesting connection between index coding problem and interference alignment technique has been appeared in \cite{MalCadJafar-CoRR12}.

In this work, we focus on the index coding problems that can be represented by a side information graph (defined in \S\ref{sec:ProbStatement}), \ie, user demands are distinct and there is exactly one receiver for each message. For this case we consider the framework for studying the index coding problem that uses ideas from graph homomorphism. More precisely, it is known that the minimum linear broadcast rate of an index coding problem of a graph $G$ over a finite field $\mathbb{F}_q$, denoted by $\lind_q(G)$, can be upper bounded by the minimum broadcast rate of another index coding problem if there exists a homomorphism from the (directed) complement of the side information graph of the first problem to that of the second problem \cite{EbrahimiJafari-ITW14}. 


Conversely, for every positive integer $k$ and prime power $q$, there exits a digraph $H^{q}_k$ such that  $\lind_q(H^q_k)$ is equal to $k$ and the complement of any digraph whose $q$-arry linear index is at most $k$ is homomorphic to $\compl{H^{q}_k}$, see \cite{EbrahimiJafari-ITW14}. The set of the graphs $H^{q}_k$ is analogous to the ``graph family $G_k$'' defined in \cite{ChlamHaviv-CoRR11} for studying a parameter of the graph called $\minrank$. In contrast to those graphs, $H^{q}_k$ are defined for arbitrary finite fields  and more importantly, they can be utilised to study the linear index code even if the graphs of interest are directed.


In this paper, we first study the structure of the digraphs $H_k^q$ defined in \cite{EbrahimiJafari-ITW14}. Analogous to the result of \cite{ChlamHaviv-CoRR11} about undirected graphs, we prove that $H_k^q$'s are vertex transitive digraphs. Using this, and by applying a lemma of Hell and Nesetril \cite{HellNese-Book04-GraphHomomorphism}, we derive a class of necessary conditions for digraphs $G$ to satisfy $\lind_q(G)\leq k$.  

In particular, we conclude that if $\lind_q(G) \leq k$ then $\frac{|G|}{\omega(G)} \leq \frac{|H_k^q|}{\omega(H_k^q)}$ in which $|G|$ and $\omega(G)$ stand for the number of vertices and the clique number of $G$, respectively.  We find a lower bound on $\omega(H_k^q)$ for every prime power $q$ and every integer $k$. Therefrom we get a new lower bound for $\lind_q(G)$. We use the same technique to obtain other lower bounds on $\lind_q(G)$ in terms of certain graph theoretic parameters of $G$ and also the underlying field size $q$.

Our next result is about the computational complexity of scalar linear index of a digraph. It is known that the deciding whether the scalar linear index of a undirected graph is equal to $k$ or not is NP-complete for $k\ge 3$ and is polynomially decidable for $k=1,2$ \cite{PeetersRene-Combinatorica96}. 
For digraphs, it is shown in \cite{DauSkaChe-IT14} that for the binary alphabet, the decision problem for $k=2$ is NP-complete. We use graph homomorphism framework to extend this result to arbitrary alphabet.



The remainder of this paper is organised as follows. In \S\ref{sec:ProbStatement} we introduce notation, some preliminary concepts about graph homomorphism and give an overview of the previous results about the connection between linear index coding and graph homomorphism. In \S\ref{sec:EquivFormLinScalBinarIndexCoding} we review one of the previous result about an equivalent formulation for the scalar linear index coding problem. The main results of this paper and their proofs are presented in \S\ref{sec:MainResult}. Finally, the paper is concluded in \S\ref{sec:Conclusion}.

\section{Notation and Problem Statement}\label{sec:ProbStatement}

\subsection{Notation and Preliminaries}
For convenience, we use $[m:n]$ to denote for the set of natural numbers $\{m,\ldots,n\}$. 
Let $x_1,\ldots,x_n$ be a set of variables. Then for any subset $\set{A}\subseteq [1:n]$ we define 
$x_{\set{A}}\triangleq (x_i:i\in\set{A})$.

All of the vectors are column vectors unless otherwise stated. 
The inner product of two vectors $v$ and $w$ is denoted by $\inner{v}{w}$.
If $A$ is a matrix, we use $[A]_j$ to denote its $j$th column.
We use $\vectm{e}_i \in \mbb{F}_q^k$ to denote for a vector which has one at $i$th
position and zero elsewhere.

A directed graph (digraph) $G$ is represented by $G(V,E)$ where $V$ is the set of vertices and $E\subseteq (V\times V)$ is the set of edges. For $v\in V(G)$ we denote by $N^+_G(v)$ as the set of outgoing neighbours of $v$, \ie, $N_G^+(v)=\{u\in V:(v,u)\in E(G)\}$. For a digraph $G$ we use $\compl{G}$ to denote for its \emph{directional complement}, \ie, $(u,v)\in E(G)$ iff $(u,v)\notin E(\compl{G})$.

A vertex coloring (coloring for short) of a graph (digraph) $G$ is an assignment of colors to the vertices of $G$ such that there is no edge between the vertices of the same color. A coloring that uses $k$ colors is called a $k$-coloring. The minimum $k$ such that a $k$-coloring exists for $G$ is called the \emph{chromatic umber}.

The \emph{clique number} of a graph (digraph) is defined as the maximum size of a subset of the vertices such that each element of the subset is connected to every other vertex in the subset by an edge (directed edge). Notice that for the directed case, between any two vertices of the subset there must be two edges in opposite directions.



The \emph{independence number} of a graph $G$, is the size of the largest independent\footnote{An \emph{independent set} of a graph $G$ is a subset of the vertices such that no two vertices in the subset represent an edge of $G$.} set and is denoted by $\alpha(G)$. For both directed and undirected graphs, it holds that $\alpha(G)=\omega(\compl{G})$.

\begin{definition}[Homomorphism, see~\cite{HellNese-Book04-GraphHomomorphism}]
Let $G$ and $H$ be any two digraphs. A \emph{homomorphism} from $G$ to $H$, written
as $\phi:G\mapsto H$ is a mapping $\phi:V(G)\mapsto V(H)$ such that $(\phi(u),\phi(v))\in E(H)$ whenever $(u,v)\in E(G)$.
If there exists a homomorphism from $G$ to $H$ we write $G\rightarrow H$,
and if there is no such homomorphism we shall write $G\nrightarrow H$. In the former case we say that $G$ is homomorphic to $H$.
We write $\Hom(G,H)$ to denote for the set of all homomorphism from $G$ to $H$.
\end{definition}

\begin{definition}
On the set of all loop-less digraphs $\set{G}$, we define the partial pre
order ``$\preccurlyeq$''  as follows. For every pair of $G,H\in\set{G}$,
$G\preccurlyeq H$ if and only if there exists a homomorphism $\phi:\compl{G}\mapsto \compl{H}$. It is straightforward to see that ``$\preccurlyeq$'' is reflexive and transitive. Moreover, if $G\preccurlyeq H$
and $H\preccurlyeq G$, then the digraphs $\compl{G}$ and $\compl{H}$ are
homomorphically equivalent (\ie, $\compl{G}\rightarrow\compl{H}$ and
$\compl{H}\rightarrow\compl{G}$). In this case we write $G\sim H$.
\end{definition}

\begin{definition}
Let $\set{D}\subseteq\set{G}$ be an arbitrary set of digraphs. A mapping $h:\set{D}\mapsto\mbb{R}$ is called \emph{increasing} over $\set{D}$ if for every 
$G,H\in\set{D}$ such that $G\preccurlyeq H$ then $h(G)\le h(H)$.
\end{definition}

\subsection{Problem Statement}
Consider the communication problem where a transmitter aims to communicate
a set of $m$ messages $x_1,\ldots,x_m\in\set{X}$ to $m$ receivers by broadcasting $k$ 
symbols $y_1,\ldots,y_{k}\in\set{Y}$, over a public noiseless channel. We assume that for each $j\in[1:m]$,
the $j$th receiver has access to the side information $x_{\set{A}_j}$, \ie, a subset
$\set{A}_j \subseteq [1:m]\setminus \{j\}$ of messages. Each receiver $j$ intends to recover $x_j$ from 
$(y_1,\ldots,y_k,x_{\set{A}_j})$.

This problem, which is a basic setting of the \emph{index coding} problem, 
can be represented by a \emph{directed side information graph} $G(V,E)$ where $V$ represents the 
set of receivers/messages and there is an edge from node $v_i$ to $v_j$, \ie, 
$(v_i,v_j)\in E$, if the $i$th receiver has $x_j$ as side information.
An index coding problem, as defined above, is completely characterized by 
the side information sets $\set{A}_j$. 
%
%

In the following definitions, we formally define validity of an index code and some other basic concepts in index coding (see also \cite{AlonLubetStavWeinHassid-Focs08,BlaKleLub-CoRR10}, and \cite{ShanmugamDimakisLangberg-CoRR13}).

\begin{definition}[Valid Index Code]
A \emph{valid index code} for $G$ over an alphabet $\set{X}$ is a set 
$(\Phi,\{\Psi_i\}_{i=1}^m)$
consisting of: (i) an encoding function $\Phi:\set{X}^m\mapsto \set{Y}^k$
which maps $m$ source messages to a
transmitted sequence of length $k$ of symbols from $\set{Y}$;
(ii) a set of $m$ decoding functions $\Psi_i$ such that for each 
$i\in[1:m]$ we have $\Psi_i(\Phi(x_1,\ldots,x_m),x_{\set{A}_i})=x_i$.
\end{definition}

\begin{definition}
Let $G$ be a digraph, and $\set{X}$ and $\set{Y}$ are the source and the message alphabets, respectively.
\\
(i) The ``\emph{broadcast rate}'' of an index code $(\Phi,\{\Psi_i\})$ is defined as $\ind_{\set{X}}(G,\Phi,\{\Psi_i\})\triangleq \frac{k \log|\set{Y}|}{\log|\set{X}|}$.\\
(ii) The ``\emph{index}'' of $G$ over $\set{X}$, denoted by $\ind_{\set{X}}(G)$, is defined as
\[
\ind_{\set{X}}(G)= \inf_{\Phi,\{\Psi_i\} } \ind_{\set{X}}(G,\Phi,\{\Psi_i\}).
\]
(iii) If $\set{X}=\set{Y}=\mbb{F}_q$ (the $q$-element finite field for some prime power $q$), the ``\emph{scalar linear index}'' of $G$, denoted by $\lind_{q}(G)$ is defined as $\lind_{q}(G) \triangleq \inf_{\Phi,\{\Psi_i\}} \ind_{\mbb{F}_q} (G,\Phi,\{\Psi_i\})$ in which the infimum is taken over the coding functions of the form $\Phi=(\Phi_1,\ldots,\Phi_k)$ and each $\Phi_i$ is a linear combination of $x_j$'s with coefficients from $\mbb{F}_q$.\\
(iv) If $\set{X}=\mbb{F}_q^\ell$ and $\set{Y}=\mbb{F}_q$, the \emph{vector linear index} for $G$, denoted by $\vlind_{q^\ell}(G)$ is defined as $\vlind_{q^\ell}(G) \triangleq \inf_{\Phi,\{\Psi_i\}} \ind_{\mbb{F}_q^\ell}(G,\Phi,\{\Psi_i\})$ where the infimum is taken over all coding functions $\Phi=(\Phi_1,\ldots,\Phi_k)$ such that $\Phi_i: \mbb{F}_q^{\ell m}\mapsto\mbb{F}_q$ are $\mbb{F}_q$-linear functions.\\
(v) The ``\emph{minimum broadcast rate}'' of the index coding problem of $G$ is defined as
$\ind(G) \triangleq \inf_{\set{X}}\ \inf_{\Phi,\{\Psi_i\} } \ind_{\set{X}}(G,\Phi,\{\Psi_i\})$.
\end{definition}

\subsection{Index Coding via Graph Homomorphism}\label{sec:IndCode_via_GrphHom}
For the sake of completeness, here in this section, we briefly review the connection of index coding problem and graph homomorphism.

Consider two different instances of the index coding problem.
It is shown in \cite{EbrahimiJafari-ITW14} that if there exists a homomorphism from the complement of the side information graph of the first problem to that of the second one, then one can translate an index code of the later problem to an index code for the former one. In particular, for the linear scalar index coding problem, one can state the following result. 
\begin{theorem}[{see \cite[Corollary~1]{EbrahimiJafari-ITW14}}]
Consider two instances of the index coding problems over the digraphs $G$ and $H$. 
If $G \preccurlyeq H$ then we have
\[
\lind_q(G) \le \lind_q(H).
\]
In other words, the function $\lind_q(\cdot)$ is a non-decreasing function on the pre order set $(G,\preccurlyeq)$.
\end{theorem}

For the case of linear scalar, it is shown in \cite{EbrahimiJafari-ITW14} that the opposite direction also holds. Namely, it is shown that for every positive integer $k$ and prime power $q$, there exits a digraph $H_k^q$ such that the $q$-arry linear index of $H_k^q$ is at most $k$ and the complement of any digraph whose $q$-arry linear index is also at most $k$ is homomorph to $\compl{H_k^q}$ .

The results of this paper are based on the properties of $H_k^q$'s and use \cite[Theorem~2]{EbrahimiJafari-ITW14} as the main tool (see also \cite{EbrahimiJafari-TechReport-CODIT14}). So, we will explain the structure of $H_k^q$ and state the theorem in the next section. 

\section{An Equivalent Formulation for Linear Scalar Index Coding Problem}\label{sec:EquivFormLinScalBinarIndexCoding}

In order to explain the structure of $H_k^q$ we need the following definitions:

\begin{definition}
A non-zero vector $a\in\mbb{F}_q^k$ is called \emph{normal} if its first non-zero element  
is equal to $1$, \ie,
\[
a = (0,\ldots,0,1,\star,\ldots,\star)^\transpose.
\]
\end{definition}

\begin{definition}
For $a\in\mbb{F}_q^k$ we define $\normop{a}$ to be the normalization of $a$, namely, $\normop{a}\triangleq \lambda a$ for some non-zero $\lambda\in\mbb{F}_q$ such that $\lambda a$ is normal.
\end{definition}

\begin{framed}
\begin{construction}[Graph $H_k^q$]\label{cons:Graph-H_k_q}
For every positive integer $k$ and prime power $q$, let $\set{V}$ be the set of all 
normal vectors in $\mbb{F}_q^k$. We define the set $\set{W}$ to be
\[
\set{W} = \left\{ (v,w) \ |\  v,w\in \set{V} \text{ and } \langle v,w \rangle \neq 0 \right\}.
\]
Now, we construct graph $H_k^q$ as follows. The vertex set of $H_k^q$ is $V(H_k^q) = \set{W}$
and $\big((v,w),(v',w') \big)\in E(H_k^q)$ if and only if $\langle v,w' \rangle \neq 0$.
In other words, $\big((v,w),(v',w') \big)\in E(\compl{H_k^q})$ if and only if $(v,w),(v',w')\in\set{W}$ 
and $\langle v,w' \rangle = 0$.
This construction leads to a graph of size $|V(H_k^q)|=\frac{q^k-1}{q-1} q^{k-1}$ which is regular with
in/out degree $q^{2(k-1)}-1$.
\end{construction} 
\end{framed}

As an example, $H_2^2$ is depicted in Figure~\ref{fig:Graph_Hk_for_K=2}.

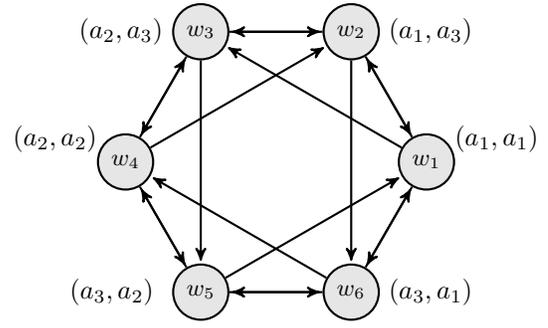
\begin{figure}
\begin{center}
\begin{tikzpicture}[->,>=stealth', shorten >=1pt, auto, node distance=2cm,
  thick, main node/.style={circle,fill=gray!20, draw, minimum size=6mm, font=\small}]

  \draw ({360/6 * (3 - 1)}:2cm) node[main node] (1) {$w_3$} [anchor=east] node {$(a_2,a_3)$~~~~};  
  \draw ({360/6 * (2 - 1)}:2cm) node[main node] (2) {$w_2$} [anchor=west] node {~~~$(a_1,a_3)$};
  \draw ({360/6 * (4 - 1)}:2cm) node[main node] (3) {$w_4$} [anchor=south east] node {$(a_2,a_2)$~~~};
  \draw ({360/6 * (5 - 1)}:2cm) node[main node] (4) {$w_5$} [anchor=east] node {$(a_3,a_2)$~~~~~};
  \draw ({360/6 * (1 - 1)}:2cm) node[main node] (5) {$w_1$} [anchor=south west] node {~~$(a_1,a_1)$};  
  \draw ({360/6 * (6 - 1)}:2cm) node[main node] (6) {$w_6$} [anchor=west] node {~~~$(a_3,a_1)$};  
  
  \path[every node/.style={font=\sffamily\small}]
    (1) edge  node {} (2)
        edge  node {} (3)
        edge node {} (4)
    (2) edge  node {} (1)
        edge  node {} (5)
        edge node {} (6)
    (3) edge  node {} (1)
        edge  node {} (4)
        edge node {} (2)
    (4) edge  node {} (3)
        edge  node {} (6)
        edge node {} (5)
    (5) edge  node {} (2)
        edge  node {} (6)
        edge node {} (1)
    (6) edge  node {} (4)
        edge  node {} (5)
        edge node {} (3);
\end{tikzpicture}
\end{center}
\caption{The digraph $H^2_2$ consists of $6$ vertices. The graph vertices are
labelled by pair of vectors $a_1=[0\ \ 1]^\transpose$, $a_2=[1\ \ 0]^\transpose$, 
and $a_3=[1\ \ 1]^\transpose$.
} 
\label{fig:Graph_Hk_for_K=2}
\end{figure}

Before we state the main tool of this paper, namely Theorem \ref{thm:Family_H_is_ICD}, we need the following definition:

\begin{definition}
A family of graphs $\{G_k\}_{k=1}^\infty$ is called \emph{$q$-index code defining ($q$-ICD)} if and only if
for every (side information) digraph $G$ we have
\[
\lind_q(G) \le k \Longleftrightarrow G \preccurlyeq G_k.
\] 
\end{definition}
It is not difficult to see that if $\{G_k\}_{k=1}^\infty$ is a $q$-ICD, so is the family $\{G'_k\}_{k=1}^\infty$ where $G_k$ and $G'_k$ are homomorphically equivalent. In particular, if $\{G_k\}_{k=1}^\infty$ is a $q$-ICD in which no $G_k$ can be replaced with a smaller digraph, then $G_k$'s are cores and conversely, if $\{G_k\}_{k=1}^\infty$ is a $q$-ICD, $\{\core(G_k)\}_{k=1}^\infty$ is the unique smallest $q$-ICD, up to isomorphism\footnote{For the definition of core of a graph, see \cite[Chapter~1.6]{HellNese-Book04-GraphHomomorphism}.}.

\begin{theorem}[see {\cite{EbrahimiJafari-TechReport-CODIT14}}] \label{thm:Family_H_is_ICD}
For every prime power $q$, the family $\{H_k^q\}_{k=1}^\infty$ as introduced above, is
$q$-ICD. 
\end{theorem}
\section{Main results}\label{sec:MainResult}
The main results of this paper are presented in this section. First, we study the properties of digraphs $H_k^q$. Base on these properties and by applying Theorem~\ref{thm:Family_H_is_ICD}, we derive two classes of lower bounds for the scalar linear index of a digraph. Finally, some computational complexity results about $\lind_q(\cdot)$ are presented.

\subsection{Properties of digraph $H_k^q$}
In this section, we state some properties of the graph
family $\{H_k^q\}_{k=1}^\infty$.

\begin{lemma}\label{lem:H_k_q-vertex_transitive}
For every $k\in\mbb{Z}^+$ and prime power $q$, digraph $H_k^q$ is vertex-transitive\footnote{A (di)graph $H$ is said to be vertex-transitive if for any vertices $u,v$ of $H$ some automorphisms of $H$ (\ie, a bijective homomorphism of $H$ to $H$) takes $u$ to $v$ (\eg, see \cite[Chapter~1, p.14]{HellNese-Book04-GraphHomomorphism}).}.
\end{lemma}
\begin{proof}
From the definition of vertex transitivity, it is enough to show that for every arbitrary vertex
$(d,e)\in V(H_k^q)$, there exists an automorphism $\phi\in\aut(H_k^q)$ such that $\phi(d,e)=(\vectm{e}_1,\vectm{e}_1)$
where $\vect{e}_1=[1\ 0\cdots 0]^\transpose\in\mbb{F}_q^k$. Let $\{\xi_1,\ldots,\xi_{k-1} \}$ be
a basis for the orthogonal complement of vector $d$, \ie, $\inner{d}{\xi_i}=0$. Since
$(d,e)\in V(H_k^q)$ it holds that $\inner{d}{e}\neq 0$ and hence $e\notin \Span(\xi_1,\ldots,\xi_{k-1})$.
Thus, $\{ \xi_1,\ldots,\xi_{k-1},d\}$ forms a basis for $\mbb{F}_q^k$. Define the $k\times k$
invertible matrix
\[
X = [e\ \ \xi_1 \ \cdots \ \xi_{k-1}].
\]
Notice that $X^\transpose d = \inner{e}{d} \vectm{e}_1$ and consequently $\normop{X^\transpose d} =\vectm{e}_1$.
Also, notice that $X \vectm{e}_1=e$ and since $X$ is invertible we can write $X^{-1} e=\vectm{e}_1$.
Define the function $\phi:V(H_k^q) \mapsto V(H_k^q)$ as 
$\phi(u,v)=(\normop{X^\transpose u},\normop{X^{-1} v} )$.
Since $(X^\transpose u)^\transpose (X^{-1} v)=\inner{u}{v}\neq 0$, we conclude that
$\phi(u,v)$ is indeed an element of $V(H_k^q)$. Last, we need to show that for each edge
$\big( (u,v),(u',v') \big)$, the image under $\phi$ is also an edge.
If $\big( (u,v),(u',v') \big)$ is an edge in $H_k^q$ then $\inner{u}{v'}=u^\transpose v' \neq 0$.
Thus, $(X^\transpose u)^\transpose (X^{-1} v') = \inner{X^\transpose u}{X^{-1} v'}\neq 0$
and therefore $\big( \phi(u,v), \phi(u',v')  \big)$ is an edge.
\end{proof}

Obviously, as a result of Lemma~\ref{lem:H_k_q-vertex_transitive}, we have the following corollary.
\begin{corollary}\label{cor:H_k_q-compl-vertex_transitive}
Digraph $\compl{H_k^q}$ is also vertex transitive.
\end{corollary}

\begin{lemma}\label{lem:H_k_q-compl-ChromaticNum}
 The chromatic number of the complement of $H_k^q$ can be upper bounded as
$\chi(\compl{H_k^q})\le \frac{q^k-1}{q-1}$.
\end{lemma}
\begin{proof}
To proof the assertion of lemma, we show that there exists a colouring of size $(q^k-1)/(q-1)$ for $\compl{H_k^q}$.
In fact we show that the vectors $d\in \mc{V}$ assigned to each vertex of $H_k^q$
are a candidate for such a colouring (see Construction~\ref{cons:Graph-H_k_q}).
First, notice that there are at most $(q^k-1)/(q-1)$ of such vectors (colours). Then we need to
show that this assignment leads to a proper colouring of $\compl{H_k^q}$.
To this end, consider two vertices $(d,e),(d,e')\in V(\compl{H_k^q})$ which have the same colour $d$.
From Construction~\ref{cons:Graph-H_k_q}, we know that $\inner{d}{e}\neq 0$ and $\inner{d}{e'}\neq 0$. 
However, the above relations mean that there is no edge in $E(\compl{H_k^q})$ from 
$(d,e)$ to $(d,e')$ and vice versa. This completes the proof of lemma.
\end{proof}

\begin{lemma}\label{lem:H_k_q-MaxCliqueSize}
 The clique number of $H_k^q$ is lower bounded as $\omega(H_k^q) = \alpha(\compl{H_k^q}) \ge \frac{1}{4} (q^2-1) q^{k-2}$.
\end{lemma}
\begin{proof}
Let $g$ be a primitive element of the field $\mbb{F}_q$ and define the sets $\set{A}$
and $\set{B}$ as follows,
\[
\set{A} = \left\{ (1, g^i, 0,\ldots,0)^\transpose \in\mbb{F}_q^k:\ 0\le i\le \frac{q-3}{2} \right\} \cup \{ \vectm{e}_1\}
\]
and
\begin{multline*}
\set{B} = \bigg\{ (1,a_2,\ldots,a_k)^\transpose \in\mbb{F}_q^k:\ a_2\notin \big\{-g^1,\ldots,-g^{\frac{q-1}{2}} \big\}, \\
a_3,\ldots,a_k\in\mbb{F}_q \bigg\}.
\end{multline*}
Notice that $|\set{A}|=\frac{q+1}{2}$ and $|\set{B}|=\frac{q-1}{2} q^{k-2}$ and also
for every $v\in\set{A}$ and $w\in\set{B}$, $\inner{v}{w}\neq 0$. Hence,
$\set{A}\times\set{B} \subseteq V(\compl{H_k^q})$ is an independent set of $\compl{H_k^q}$
of size $\frac{(q^2-1)}{4} q^{k-2}$. Thus $\omega(H_k^q) = \alpha(\compl{H_k^q}) \ge \frac{(q^2-1)}{4} q^{k-2}$. 
\end{proof}

\subsection{Lower Bounds}
As a result of Theorem~\ref{thm:Family_H_is_ICD}, we can prove the following lemma.

\begin{lemma}\label{lem:GenLowerBoundUsingMaximalElement}
Suppose that $h$ is an increasing function on \mbox{$(\set{G},\preccurlyeq)$} and $r$ is an upper bound on $h(H_k^q)$. For every digraph $G$, if $h(G)>r$ then $\lind_q(G) > k$.
\end{lemma}
\begin{proof}
If $\lind_q(G)\le k$ then by Theorem~\ref{thm:Family_H_is_ICD}, $G\preccurlyeq H_k^q$ 
and therefore $h(G)\le h(H_k^q)\le r$ which is a contradiction.
\end{proof}

Lemma~\ref{lem:GenLowerBoundUsingMaximalElement} is a powerful tool to find lower bounds on the scalar linear index of an index coding problem. Actually for every increasing function $h$ on $(\set{G},\preccurlyeq)$ we have one lower bound on the index coding problem. In the next theorem, using Lemma~\ref{lem:GenLowerBoundUsingMaximalElement}, we provide an alternative proof for the known lower bound on $\lind_q(G)$ in terms of the chromatic number of $\compl{G}$ \cite{LangbergSprinston-ISIT08}.

\begin{theorem}\label{thm:Ind_LowerBound_ChromaticNumber}
For every digraph $G$ we have 
\[
\lind_q(G) \ge \log_q\left( (q-1)\chi(\compl{G}) + 1 \right).
\]
\end{theorem}
\begin{proof}
The function $h(G)=\chi(\compl{G})$ is an increasing function on $(\set{G},\preccurlyeq)$ (\eg, see \cite[Corollary~1.8]{HellNese-Book04-GraphHomomorphism}). Suppose that $\lind_q(G)=k$. Therefore, by Theorem~\ref{thm:Family_H_is_ICD}, we have $G\preccurlyeq H_k^q$. Thus $\chi(\compl{G}) \leq \chi(\compl{H^q_k}) \leq \frac{q^k-1}{q-1}$ where the last inequality follows from Lemma~\ref{lem:H_k_q-compl-ChromaticNum}.
This completes the proof of theorem.
\end{proof}


In the following we introduce another tool to find lower bounds on $\lind_q(G)$.
Hell and Nesetril in \cite[Proposition~1.22]{HellNese-Book04-GraphHomomorphism} showed that if a graph $H$
is vertex transitive and $\phi:G\rightarrow H$ is a homomorphism, then for every graph $K$
we have $\frac{|G|}{N(G,K)} \le \frac{|H|}{N(H,K)}$
in which $N(L,K)$ is the size of the largest induced subgraph $L'$ of $L$ such that there
exists a homomorphism from $L'$ to $K$. Even though this result is stated for undirected
graphs, its proof naturally extends to digraphs. Therefore, using this result and 
Theorem~\ref{thm:Family_H_is_ICD} we can deduce the following theorem.

\begin{theorem}\label{thm:N(G,K)_N(H,K)_LowerBound}
For every digraph $G$, if there exists a digraph $K$such that
\[
\frac{|G|}{N(\compl{G},K)} > \frac{|H_k^q|}{N(\compl{H_k^q},K)}
\]
then $\lind_q(G) \ge k+1$.
\end{theorem}
\begin{proof}
We prove the theorem by contradiction.
Suppose that $\lind_q(G)\le k$. Then by Theorem~\ref{thm:Family_H_is_ICD}, we have
$\Hom(\compl{G},\compl{H_k^q})\neq\varnothing$.
Also by Corollary~\ref{cor:H_k_q-compl-vertex_transitive}, $\compl{H_k^q}$ is a vertex transitive
digraph. Now we can apply the result of \cite[Proposition~1.22]{HellNese-Book04-GraphHomomorphism} to conclude that for every digraph $K$ we have
$\frac{|G|}{N(\compl{G},K)} \le \frac{|H_k^q|}{N(\compl{H_k^q},K)}$ which is a contradiction.
Thus we are done.
\end{proof}

An interesting special case of Theorem~\ref{thm:N(G,K)_N(H,K)_LowerBound} is when $K$
is a complete graph with $l$ vertices. In this case $N(G,K)$ is equal to the size of the
largest $l$-colourable induced subgraph of $G$. In particular the following corollary
holds.
\begin{corollary} \label{corollary2}
For every digraph $G$ if 
\[
\frac{|G|}{\omega(G)} > \frac{4(q^k-1)q}{(q^2-1)(q-1)}
\]
then $\lind_q(G) \ge k+1$.
In other words we can state the following lower bound
\[
\lind_q(G) \ge \log_q\left[1 + \frac{(q^2-1)(q-1)}{4q} \frac{|G|}{\omega(G)} \right].
\]
\end{corollary}
\begin{proof}
In Theorem~\ref{thm:N(G,K)_N(H,K)_LowerBound} set $K$ to be $K_1$. Then
$\omega(G)=\alpha(\compl{G})= N(\compl{G},K_1)$. Now Lemma~\ref{lem:H_k_q-MaxCliqueSize}
completes the proof.
\end{proof}

\begin{example}
Let $G$ be a directed graph with $n$ vertices so that for every pair $(i,j)$ of the vertices either $(i,j)\in E(G)$ or $(j,i)\in E(G)$ but not both. In this case, $\alpha(G)=\omega(G)=1$. The lower bound derived from Corollary~\ref{corollary2} is 
\[
\lind_q(G) \ge \log_q\left[1 + \frac{(q^2-1)(q-1)}{4q} n \right].
\]
This lower bound is slightly better than the lower bound obtained from Theorem~\ref{thm:Ind_LowerBound_ChromaticNumber} and is significantly stronger that the trivial bound $\lind_q(G)\ge \alpha(G)$.
\end{example}

Another special case is when $K$ is equal to the complete undirected graph on $l\geq 2$ vertices, i.e. $K=K_l$. In this case, $N(\compl{G},K)$ is the size of the largest $l$-colorable induced subgraph of $\compl{G}$. For example, $N(\compl{G},K_2)$ is the size of the largest induced bipartite of $\compl{G}$. For these choices of $K$, we are able to find a lower bound on $N(\compl{H_k^q},K)$ and therefore we are able to find an upper bound on $\frac{|G|}{N(\compl{G},K)}$.

The idea of finding a large induced subgraph of $\compl{H_k^q}$ that is $l$-colorable is similar to that of finding a large independent set in $\compl{H_k^q},K$. Formally we have the following result:

\begin{lemma}\label{lem:largeinduced}
let $1<l <k-1$ be a positive integer number. Then
\begin{align*}
N(\compl{H_k^q},K_l) &\geq \sum_{i=1}^{l}{\frac{1}{4} (q^2-1)(q^{k-1-i})}\\
&= \frac{1}{4} (q+1)(q^{l}-1)(q^{k-l-1}).
\end{align*}
\end{lemma} 
\begin{proof}
Similar to the proof of Lemma~\ref{lem:H_k_q-MaxCliqueSize}, Let $g$ be a primitive element of the field $\mbb{F}_q$ and for every $j=1,2,\ldots,l$ define the sets $\set{A}_j$
and $\set{B}_j$ as follows,
\begin{multline*}
\set{A}_j = \left\{ (0,0,\ldots,0,1, g^i, 0,\ldots,0)^\transpose \in\mbb{F}_q^k:\ 0\le i\le \frac{q-3}{2} \right\}\\ \cup \{ \vectm{e}_{j}\}
\end{multline*}
in which the number of the leading zeros are equal to $j-1$ 
and
\begin{multline*}
\set{B}_j = \bigg\{ (0,0,\ldots,0,1,a_{j+1},\ldots,a_k)^\transpose \in\mbb{F}_q^k:\\ a_{j+1}\notin \big\{-g^1,\ldots,-g^{\frac{q-1}{2}} \big\}, 
a_{j+2},\ldots,a_k\in\mbb{F}_q \bigg\}.
\end{multline*}
Notice that $|\set{A}_j|=\frac{q+1}{2}$ and $|\set{B}_j|=\frac{q-1}{2} q^{k-j-1}$ and also
for every $v\in\set{A}_j$ and $w\in\set{B}_j$, $\inner{v}{w}\neq 0$. Hence,
$\set{A}_j\times\set{B}_j \subseteq V(\compl{H_k^q})$ are disjoint independent sets of $\compl{H_k^q}$
of size $\frac{(q^2-1)}{4} q^{k-j-1}$. The union of  $\set{V}_j=\set{A}_j\times\set{B}_j$ together with the edges between different $\set{V}_i$'s form an $l$-colorable induced subgraph of $\compl{H_k^q}$ of desired size.
\end{proof}

\begin{corollary}
For every digraph $G$, if $\lind_q(G)=k$  and $l<k-1$ then $$ \frac{|G|}{N(\compl{G},K_l)} \leq  \frac{4q^l(q^k-1)}{ (q^2-1)(q^l-1)}.$$
Equivalently, 

\[
\lind_q(G) \ge \log_q\left[1 + \frac{(q^2-1)(q^l-1)}{4q^l} \frac{|G|}{N(\compl{G},K_l)} \right].
\]
\end{corollary}

The proof of the previous Corollary is very similar to that of Corollary \ref{corollary2}.

\subsection{Computational Complexity of $\lind_q(\cdot)$}
By applying the result of \cite{BartoKozikNiven-SIAM09} that proves the conjecture
\cite[Conjecture~5.16]{HellNese-Book04-GraphHomomorphism}
and using Theorem~\ref{thm:Family_H_is_ICD}, we show that the decision problem
of $\lind_q(G)\le k$ is NP-complete for $k\ge 2$.

First, let us state the following theorem.
\begin{theorem}[{see \cite{BartoKozikNiven-SIAM09} and \cite[Conjecture~5.16]{HellNese-Book04-GraphHomomorphism}}] \label{thm:If_H_no_directed_cycle_NP_complete}
Suppose that $H$ is a digraph with no vertex of in degree or out degree one. If core of $H$ is
not a directed cycle the following decision problem is NP-complete.
\begin{framed}
\begin{tabular}{r p{6cm}}
Input: & A digraph $G$\\
Output: & If $\hom(G,H)\neq \varnothing$ then return ``yes,'' else return ``no.''
\end{tabular}
\end{framed}
\end{theorem}

We will apply Theorem~\ref{thm:If_H_no_directed_cycle_NP_complete} to digraphs $H_k^q$
for every $k\ge 2$.

\begin{theorem}\label{thm:lind_q-decision-problem-NP_Complete}
For every finite field $\mbb{F}_q$ and every $k\ge 2$ the following decision problem is 
NP-complete.
\begin{framed}
\begin{tabular}{r p{6cm}}
Input: & A digraph $G$\\
Output: & If $\lind_q(G)\le k$ return ``yes,'' otherwise return ``no.''
\end{tabular}
\end{framed}
\end{theorem}
\begin{proof}
Based on Theorem~\ref{thm:Family_H_is_ICD} and Theorem~\ref{thm:If_H_no_directed_cycle_NP_complete},
it is sufficient to show that $\compl{H_k^q}$ has no vertex of in/out degree one and also
its core is not a directed cycle. From Construction~\ref{cons:Graph-H_k_q} we know 
that for $k\ge 2$, $\compl{H_k^q}$ has
no source or sink. It remains to show that the core of $\compl{H_k^q}$ for any 
prime power $q$ and any integer $k$, is not a directed cycle.

To this end, consider the subgraph $D$ of $\compl{H_k^q}$ depicted in 
Figure~\ref{fig:Induced-Subgraph-of-H_k_q-Compl}.
\begin{figure}
\begin{center}
\begin{tikzpicture}[->,>=stealth', shorten >=1pt, auto, node distance=2cm,
  thick, main node/.style={circle,fill=gray!20, draw, minimum size=4mm, font=\small}]

  \draw (0,0) node[main node] (1) {} [anchor=west] node {~~$(\vectm{e}_2, \vectm{e}_1-\vectm{e}_2)$};
  \draw (1,1.5) node[main node] (2) {} [anchor=west] node {~~$(\vectm{e}_1+\vectm{e}_2, \vectm{e}_2)$};
  \draw (-1,1.5) node[main node] (3) {} [anchor=east] node {$(\vectm{e}_1, \vectm{e}_1)$~~~};
  \draw (-2,0) node[main node] (4) {} [anchor=east] node {$(\vectm{e}_2, \vectm{e}_2)$~~};

  \path[every node/.style={font=\sffamily\small}]
	(1) edge (3)
	(3)	 edge (2)
	(2)	 edge (1)
	(4) edge [bend left] (3)
	(3) edge [bend left] (4);
\end{tikzpicture}
\end{center}
\caption{The induced subgraph $D$ of $\compl{H_k^q}$ which is used in the 
proof of Theorem~\ref{thm:lind_q-decision-problem-NP_Complete} 
to show that the core of $\compl{H_k^q}$ cannot be
a directed cycle.} \label{fig:Induced-Subgraph-of-H_k_q-Compl}
\end{figure}
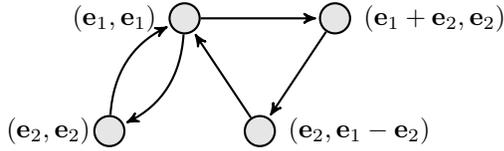
Clearly from this induced subgraph $D$ there is no homomorphism to any
directed cycle. Thus the core of $\compl{H_k^q}$ cannot be a directed cycle.
\end{proof}

Notice that the NP-completeness of the decision problem $\lind_q(G)\le k$
was known for $q=2$ and $k\ge 3$, see \cite{PeetersRene-Combinatorica96,ChlamHaviv-CoRR11}. In fact this
decision problem is NP-complete even when $G$ is an undirected graph \cite{PeetersRene-Combinatorica96}.
However, for undirected graphs, the decision problem $\lind_q(G)\le 2$ is
polynomially solvable since for undirected graphs we have $\lind_q(G)\le 2$
if and only if $G$ is complement of a bipartite graph. Theorem~\ref{thm:lind_q-decision-problem-NP_Complete} shows
that for digraphs, even the decision problem $\lind_q(G)\le 2$ is hard. 

\section{Conclusion}\label{sec:Conclusion}
It is previously shown that scalar linear index of a digraph $G$ over $\mbb{F}_q$ is less than ot equal to $k$ if there exists a graph homomorphism from the complement of $G$ to the complement of certain graphs denoted by $H_k^q$ \cite{EbrahimiJafari-ITW14} (see also \cite{ChlamHaviv-CoRR11} for the special case of undirected graphs and $\mbb{F}_2$).

In this work, we investigated the structure and properties of digraphs $H_k^q$. In particular, similar to the result of \cite{ChlamHaviv-CoRR11} for undirected graphs, we show that $H_k^q$ are vertex transitive. Using this, and by applying a result from Hell and Nesetril \cite{HellNese-Book04-GraphHomomorphism}, we derive a class of necessary conditions for digraph $G$ to satisfy $\lind_q(G)\le k$.

Based on the structure of graphs $H_k^q$, our next result is about the computational complexity of $\lind_q(G)$.
It is known that the deciding whether the scalar linear index of a undirected graph is equal to $k$ or not is NP-complete for $k\ge 3$ and is polynomially decidable for $k=1,2$ \cite{PeetersRene-Combinatorica96}. However, using the graph homomorphism framework, we show that this decision problem is NP-complete even for $k=2$ if we consider digraphs.

\section*{Acknowledgement}

The authors would like to thank Dr. P.~Hell, Dr. S.~Jaggi and Dr. A.~Rafiey for their invaluable comments.

\bibliographystyle{IEEEtran}
\bibliography{../index_coding}

\begin{thebibliography}{10}
\providecommand{\url}[1]{#1}
\csname url@samestyle\endcsname
\providecommand{\newblock}{\relax}
\providecommand{\bibinfo}[2]{#2}
\providecommand{\BIBentrySTDinterwordspacing}{\spaceskip=0pt\relax}
\providecommand{\BIBentryALTinterwordstretchfactor}{4}
\providecommand{\BIBentryALTinterwordspacing}{\spaceskip=\fontdimen2\font plus
\BIBentryALTinterwordstretchfactor\fontdimen3\font minus
  \fontdimen4\font\relax}
\providecommand{\BIBforeignlanguage}[2]{{%
\expandafter\ifx\csname l@#1\endcsname\relax
\typeout{** WARNING: IEEEtran.bst: No hyphenation pattern has been}%
\typeout{** loaded for the language `#1'. Using the pattern for}%
\typeout{** the default language instead.}%
\else
\language=\csname l@#1\endcsname
\fi
#2}}
\providecommand{\BIBdecl}{\relax}
\BIBdecl

\bibitem{YosBirkJayKol-IT11}
Z.~Bar-Yossef, Y.~Birk, T.~S. Jayram, and T.~Kol, ``Index coding with side
  information,'' \emph{IEEE Trans. Inf. Theory}, vol.~57, no.~3, pp.
  1479--1494, 2011.

\bibitem{ChlamHaviv-CoRR11}
E.~Chlamtac and I.~Haviv, ``Linear index coding via semidefinite programming,''
  \emph{CoRR}, vol. abs/1107.1958, 2011.

\bibitem{PeetersRene-Combinatorica96}
R.~Peeters, ``Orthogonal representations over finite fields and the chromatic
  number of graphs,'' \emph{Combinatorica}, vol.~16, no.~3, pp. 417--431, 1996.

\bibitem{EbrahimiJafari-ITW14}
J.~Ebrahimi~Boroojeni and M.~Jafari~Siavoshani, ``On index coding and graph
  homomorphism,'' in \emph{IEEE Information Theory Workshop}, 2014.

\bibitem{HellNese-Book04-GraphHomomorphism}
P.~Hell and J.~Nesetril, \emph{Graphs and Homomorphisms}, ser. Oxford Lecture
  Series in Mathematics and Its Applications.\hskip 1em plus 0.5em minus
  0.4em\relax OUP Oxford, 2004.

\bibitem{DauSkaChe-IT14}
S.~H. Dau, V.~Skachek, and Y.~M. Chee, ``Optimal index codes with near-extreme
  rates,'' \emph{Information Theory, IEEE Transactions on}, vol.~60, no.~3, pp.
  1515--1527, 2014.

\bibitem{BirkKol-IndexCoding-INFOCOM98}
Y.~Birk and T.~Kol, ``Informed-source coding-on-demand (iscod) over broadcast
  channels,'' in \emph{in Proc. 17th Ann. IEEE Int. Conf. Comput. Commun.
  (INFOCOM)}, 1998, pp. 1257--1264.

\bibitem{AlonLubetStavWeinHassid-Focs08}
N.~Alon, E.~Lubetzky, U.~Stav, A.~Weinstein, and A.~Hassidim, ``Broadcasting
  with side information,'' in \emph{Proceedings of the 2008 49th Annual IEEE
  Symposium on Foundations of Computer Science}, ser. FOCS '08, 2008, pp.
  823--832.

\bibitem{LubetStav-IT09}
E.~Lubetzky and U.~Stav, ``Nonlinear index coding outperforming the linear
  optimum,'' \emph{IEEE Trans. Inf. Theory}, vol.~55, no.~8, pp. 3544--3551,
  2009.

\bibitem{RouaySprintGeorgh-IT10}
S.~El~Rouayheb, A.~Sprintson, and C.~Georghiades, ``On the index coding problem
  and its relation to network coding and matroid theory,'' \emph{IEEE Trans.
  Inf. Theory}, vol.~56, no.~7, pp. 3187--3195, 2010.

\bibitem{BlaKleLub-CoRR10}
A.~Blasiak, R.~D. Kleinberg, and E.~Lubetzky, ``Index coding via linear
  programming,'' \emph{CoRR}, vol. abs/1004.1379, 2010.

\bibitem{BerLangberg-ISIT11}
Y.~Berliner and M.~Langberg, ``Index coding with outerplanar side
  information,'' in \emph{IEEE Int. Symp. Inf. Theory}, 2011, pp. 806--810.

\bibitem{HavivLangberg-ISIT12}
I.~Haviv and M.~Langberg, ``On linear index coding for random graphs,'' in
  \emph{IEEE Int. Symp. Inf. Theory}, 2012, pp. 2231--2235.

\bibitem{TehraniDimakisNeely-ISIT12}
A.~Tehrani, A.~Dimakis, and M.~Neely, ``Bipartite index coding,'' in \emph{IEEE
  Int. Symp. Inf. Theory}, 2012, pp. 2246--2250.

\bibitem{MalCadJafar-CoRR12}
H.~Maleki, V.~R. Cadambe, and S.~A. Jafar, ``Index coding - an interference
  alignment perspective,'' \emph{CoRR}, vol. abs/1205.1483, 2012.

\bibitem{ShanmugamDimakisLangberg-CoRR13}
K.~Shanmugam, A.~G. Dimakis, and M.~Langberg, ``Local graph coloring and index
  coding,'' \emph{CoRR}, vol. abs/1301.5359, 2013.

\bibitem{ArbabBanderKimSasogluWang-ISIT13}
F.~Arbabjolfaei, B.~Bandemer, Y.-H. Kim, E.~Sasoglu, and L.~Wang, ``On the
  capacity region for index coding,'' in \emph{IEEE Int. Symp. Inf. Theory},
  2013, pp. 962--966.

\bibitem{EffrosSalimLangberg-CoRR12}
M.~Effros, S.~Y.~E. Rouayheb, and M.~Langberg, ``An equivalence between network
  coding and index coding,'' \emph{CoRR}, vol. abs/1211.6660, 2012.

\bibitem{EbrahimiJafari-TechReport-CODIT14}
\BIBentryALTinterwordspacing
J.~Ebrahimi~Boroojeni and M.~Jafari~Siavoshani, ``On index coding and graph
  homomorphism,'' \emph{Technical Report}, 2014. [Online]. Available:
  \url{http://mahdi.jafaris.net/download/TechReports/IndexCodingv2.pdf}
\BIBentrySTDinterwordspacing

\bibitem{LangbergSprinston-ISIT08}
M.~Langberg and A.~Sprintson, ``On the hardness of approximating the network
  coding capacity,'' in \emph{IEEE International Symposium on Information
  Theory}, 2008, pp. 315--319.

\bibitem{BartoKozikNiven-SIAM09}
L.~Barto, M.~Kozik, and T.~Niven, ``The csp dichotomy holds for digraphs with
  no sources and no sinks (a positive answer to a conjecture of bang-jensen and
  hell),'' \emph{SIAM Journal on Computing}, vol.~38, no.~5, pp. 1782--1802,
  2009.

\end{thebibliography}
\end{document}